\documentclass[a4paper,UKenglish]{lipics-v2016}
\usepackage{microtype}%if unwanted, comment out or use option "draft"

\title{Self-Stabilizing Supervised Publish-Subscribe Systems\footnote{This work was partially supported by the German Research Foundation (DFG) within the Collaborative Research Center "On-The-Fly Computing" (SFB 901).}}
\titlerunning{Self-Stabilizing Supervised Publish-Subscribe Systems} %optional, in case that the title is too long; the running title should fit into the top page column

\author[1]{Michael Feldmann}
\author[2]{Christina Kolb}
\author[3]{Christian Scheideler}
\author[4]{Thim Strothmann}
\affil[1]{Paderborn University, F\"urstenallee 11, 33102 Paderborn, Germany\\
  \texttt{mfeldma2@mail.upb.de}}
\affil[2]{Paderborn University, F\"urstenallee 11, 33102 Paderborn, Germany\\
  \texttt{ckolb@mail.upb.de}}
\affil[3]{Paderborn University, F\"urstenallee 11, 33102 Paderborn, Germany\\
  \texttt{scheideler@upb.de}}
\affil[4]{Paderborn University, F\"urstenallee 11, 33102 Paderborn, Germany\\
  \texttt{thim@mail.upb.de}}
\authorrunning{M. Feldmann, C. Kolb, C. Scheideler and T. Strothmann} %mandatory. First: Use abbreviated first/middle names. Second (only in severe cases): Use first author plus 'et. al.'

\Copyright{Michael Feldmann, Christina Kolb, Christian Scheideler and Thim Strothmann}%mandatory, please use full first names. LIPIcs license is "CC-BY";  http://creativecommons.org/licenses/by/3.0/

\subjclass{C.2.4 Distributed Systems}% mandatory: Please choose ACM 1998 classifications from http://www.acm.org/about/class/ccs98-html . E.g., cite as "F.1.1 Models of Computation". 
\keywords{Topological Self-stabilization, Publish-Subscribe, Supervised Overlay}% mandatory: Please provide 1-5 keywords
\EventEditors{}
\EventNoEds{2}
\EventLongTitle{}
\EventShortTitle{}
\EventAcronym{}
\EventYear{2017}
\EventDate{July 11, 2017}
\EventLocation{Little Whinging, United Kingdom}
\EventLogo{}
\SeriesVolume{1}
\ArticleNo{1}
\usepackage[noend]{algpseudocode}
\usepackage{algorithm}
\theoremstyle{plain}
\newtheorem{claim}[theorem]{Claim}
\DeclareMathOperator*{\argmax}{argmax} % no space, limits on side in displays
\begin{document}

\maketitle

\begin{abstract}
In this paper we present two major results:
First, we introduce the first self-stabilizing version of a supervised overlay network (as introduced in~\cite{DBLP:conf/ispan/KothapalliS05}) by presenting a self-stabilizing supervised skip ring.
Secondly, we show how to use the self-stabilizing supervised skip ring to construct an efficient self-stabilizing publish-subscribe system.
%Our publish-subscribe system provably constructs a supervised overlay for %each topic if the initial state is weakly connected.
That is, in addition to stabilizing the overlay network, every subscriber of a topic will eventually know all of the publications that have been issued so far for that topic. The communication work needed to processes a subscribe or unsubscribe operation is just a constant in a legitimate state, and the communication work of checking whether the system is still in a legitimate state is just a constant on expectation for the supervisor as well as any process in the system.
\end{abstract}

\section{Introduction}
The publish subscribe paradigm (\cite{DBLP:journals/csur/EugsterFGK03, DBLP:conf/sigmod/FabretJLPRS01}) is a very popular paradigm for the targeted dissemination of information. 
It allows clients to subscribe to certain topics or contents so that they will only receive information that matches their interests. 
In the traditional client-server approach the dissemination of information is handled by a server (also called broker), which has the benefit that the publishers are decoupled from the subscribers: the publisher does not have to know the relevant subscribers, and the publisher and subscribers do not have to be online at the same time. 
However, in this case the availability of the publish subscribe system critically depends on the availability of the server, and the server has to be powerful enough to handle the dissemination of the publish requests. 
An alternative approach is to use a peer-to-peer system. 
However, if no commonly known gateway is available, the peer-to-peer system cannot recover from overlay network partitions. 
In practice, peer-to-peer systems usually have a commonly known gateway since otherwise new peers may not be able to get in contact with a peer that is currently in the system (and can therefore process the join request). 
In our supervised overlay network approach we assume that there is a commonly known gateway, called \emph{supervisor}, that just handles subscribe and unsubscribe requests but does not handle the dissemination of publish requests, which will be handled by the subscribers in a peer-to-peer manner.
We are interested in realizing a \emph{topic-based} supervised publish subscribe system, which means that peers can subscribe to certain topics (that are usually relatively broad and predefined by the supervisor).

Topic-based publish subscribe systems have many important applications. 
Apart from providing a targeted news service, they can be used, for example, to realize a group communication service \cite{DBLP:journals/tocs/FeketeLS01}, which is considered an important building block for many other applications ranging from chat groups and collaborative working groups to online market places (where clients publish service requests), distributed file systems or transaction systems. 
To ensure the reliable dissemination of publish requests in a topic-based publish subscribe system, we present a self-stabilizing supervised publish subscribe system, which ensures that for any initial state (including overlay network partitions) eventually a legitimate state will be reached in which all subscribers of a topic know about all publish requests that have been issued for that topic. 
We also show that the overhead for the supervisor in our system is very low.
In fact, the message overhead of the supervisor is just a constant for subscribe and unsubscribe operations, and the supervisor has a low maintenance overhead in a legitimate state.

\subsection{Model}
We model the overlay network of a distributed system as a directed graph $G=(V,E)$, where $n = |V|$. Each peer is represented by a node $v \in V$. Each node $v \in V$ is identified by its unique \emph{ reference} or \emph{identifier}  $v.id \in \mathbb{N}$ (called \textit{ID}). Additionally, each node $v$ maintains local protocol-based variables and has a \textit{channel} $v.Ch$, which is a system-based variable that contains incoming messages. We assume a channel to be able to store any finite number of messages, and messages are never duplicated or get lost in the channel.
If a node $u$ has the reference of some other node $v$, $u$ can send a message $m$ to $v$ by putting $m$ into $v.Ch$.
There is a directed edge $(u,v) \in E$ whenever $u$ stores a reference of $v$ in its local memory or there is a message in $u.Ch$ carrying the reference of $v$. In the former case, we call that edge \textit{explicit} and in the latter case we call that edge \textit{implicit}. Note that every node is assumed to know the supervisor, and this information is read-only, so $G$ always contains a directed star graph from all peers to the supervisor.

Nodes may execute \textit{actions}: An action is just a standard procedure and has the form $\langle label \rangle(\langle parameters \rangle):\langle command \rangle$, where $label$ is the name of that action, $parameters$ defines the set of parameters and $command$ defines the statements that are executed when calling that action.
It may be called locally or remotely, i.e., every message that is sent to a node has the form $\langle label \rangle(\langle parameters \rangle)$.
When a node $u$ processes a message $m$, then $m$ is removed from $u.Ch$.
Additionally, there is an action that is not triggered by messages but is executed periodically by each node.
We call this action \textsc{Timeout}.

We define the \textit{system state} to be an assignment of a value to every node's variables and messages to each channel.
A \textit{computation} is an infinite sequence of system states, where the state $s_{i+1}$ can be reached from its previous state $s_i$ by executing an action that is enabled in $s_i$.
We call the first state of a given computation the \textit{initial state}.
We assume \textit{fair message receipt}, meaning that every message of the form $\langle label \rangle(\langle parameters \rangle)$ that is contained in some channel, is eventually processed.
Furthermore, we assume \textit{weakly fair action execution}, meaning that if an action is enabled in all but finitely many states of a computation, then this action is executed infinitely often.
Consider the \textsc{Timeout} action as an example for this.
We place no bounds on message propagation delay or relative node execution speed, i.e.,  we allow fully asynchronous computations and non-FIFO message delivery.
Our protocol does not manipulate node identifiers and thus only operates on them in \emph{compare-store-send} mode, i.e., the nodes are only allowed to compare node IDs, store them in a node's local memory or send them in a message.

In this paper we assume for simplicity that there are no corrupted IDs (i.e., IDs of unavailable nodes) in the initial state of the system. However, dealing with them is easy when having a failure detector that is eventually correct since, due to the supervisor, the correctness of our protocol cannot be endangered by sending messages to non-available nodes. Since our protocol just deals with IDs in a compare-store-send manner, this implies that node IDs will always be non-corrupted for all computations.
Nevertheless, the node channels may initially contain an arbitrary finite number of messages containing false information. We call these messages \textit{corrupted}, and we will argue that eventually there will not be any corrupted messages in the system. 
We will show that our protocol realizes a self-stabilizing supervised publish-subscribe system.

\begin{definition}[Self-stabilization] \label{def:self_stabilization}
A protocol is \emph{self-stabilizing} w.r.t. a set of legitimate states if it satisfies the following two properties:
\begin{itemize}
	\item[-] Convergence: Starting from an arbitrary system state, the protocol is guaranteed to arrive at a legitimate state.
	\item[-] Closure: Starting from a legitimate state, the protocol remains in legitimate states thereafter.
\end{itemize}
\end{definition}

\subsection{Related Work}
The concept of self-stabilizing algorithms for distributed systems goes back to the year 1974, when E. W. Dijkstra introduced the idea of self-stabilization in a token-based ring \cite{DBLP:journals/cacm/Dijkstra74}.
Many self-stabilizing protocols for various types of overlays have been proposed, like sorted lists~\cite{DBLP:conf/alenex/OnusRS07}, de Bruijn graphs~\cite{DBLP:conf/sss/RichaSS11}, Chord graphs~\cite{DBLP:journals/mst/KniesburgesKS14} and many more. %, Chord graphs~\cite{DBLP:journals/mst/KniesburgesKS14}, Skip graphs~\cite{DBLP:conf/podc/JacobRSST09} and many more.
There is even a universal approach, which is able to derive self-stabilizing protocols for several types of topologies \cite{DBLP:journals/tcs/BernsGP13}.

The cycle topology is particularly important for our work.
Our cycle protocol is based on \cite{DBLP:conf/ipps/KniesburgesKS12}, in which the authors construct a self-stabilizing cycle that acts as a base for additional long-range links, both together forming a small-world network.

The paper closest to our work is by Kothapalli and Scheideler~\cite{DBLP:conf/ispan/KothapalliS05}.
The authors provide a general framework for constructing a supervised peer-to-peer system in which the supervisor only has to store a constant amount of information about the system at any time and only has to send out a constant number of messages to integrate or remove a node.
However, their system is not self-stabilizing.

In the literature there are publish-subscribe systems that are self-stabilizing: e.g. in~\cite{DBLP:conf/europar/MuhlJHWUF05} the authors present different content-based routing algorithms in a self-stabilizing (acyclic) broker overlay network that clients can publish messages to.
Their main idea is a leasing mechanism for routing tables such that it is guaranteed that once a client subscribes to a topic there is a point in time such that every publication which is issued thereafter is delivered to the newly subscribed client (i.e., there are no guarantees for older publications).
While the authors focus on the routing tables and take the overlay network as a given ingredient, our work focuses on constructing a self-stabilizing supervised overlay network and then using it to obtain a self-stabilizing publish-subscribe system.

A self-stabilizing publish-subscribe system for wireless ad-hoc networks is proposed in~\cite{DBLP:conf/sss/SiegemundT16}, which builds upon the work of~\cite{DBLP:conf/kivs/SiegemundTM15}: Similar to our work, the authors arrange nodes in a cycle with shortcuts and present a routing algorithm that makes use of these shortcuts to deliver new publications for topics to subscribers only after $\mathcal O(n)$ steps.
Subscribe and unsubscribe requests are processed by updating the routing table at nodes.
Both systems described above differ from our approach, as they solely focus on the routing scheme and updates of the routing tables, while we focus on updating the topology upon subscribe/unsubscribe requests.
Additionally, our system is able to deliver publications in $\mathcal O(\log n)$ steps, if we use flooding, since we use a network with logarithmic diameter.
Furthermore, we are also able to deliver all publications of a domain to a new subscriber after only a constant number of rounds.

There is a close relationship between group communication services (e.g.,~\cite{DBLP:journals/tocs/FeketeLS01, DBLP:journals/tdsc/AmirNST05}) and publish-subscribe systems.
Processes are ordered in groups in both paradigms and group-messages are only distributed among all members of some group.
Self-stabilizing group communication services are proposed in~\cite{DBLP:journals/tmc/DolevSW06} for ad-hoc networks and in~\cite{DBLP:journals/acta/DolevS04} for directed networks.
However, there are some key differences: In group communication services, participants have to agree on group membership views.
This results in a high memory overhead for each member of a group, as nodes in a group technically form a clique.
On the other hand subscribers of topics in publish-subscribe systems are in general not interested in any other members of the topic.
For our approach, this results in logarithmic worst-case and constant average case degree for subscribers.

\subsection{Our Contribution}
To the best of our knowledge, we present the first self-stabilizing protocol for a supervised overlay network.
We focus on a topology that is a ring with shortcuts which we call \emph{skip ring}.
The corresponding protocol \textsc{BuildSR} is split up into two subprotocols: One protocol is executed at the supervisor (see Section~\ref{sec:supervisor}), 
%to let it maintain a $database$ containing all subscribers.
the other one is executed by each subscriber (Section~\ref{sec:subscriber}).
%, s.t. subscribers will eventually form a sorted ring with shortcuts.
Our basic protocol assumes that all references actually belong to existing nodes.
However, we also present an extension (see Section~\ref{sec:subscriber:failures}) to handle references to non-existing nodes and unannounced failures of nodes.
In contrast to the supervised overlay network proposed in~\cite{DBLP:conf/ispan/KothapalliS05}, our new protocol lets the supervisor handle multiple insertions/deletions in parallel without having to rely on confirmations from other nodes, however, at the cost of storing much more references than the solution in~\cite{DBLP:conf/ispan/KothapalliS05}.

The skip ring shares some similarities with other shortcut-based peer-to-peer systems like Chord networks~\cite{DBLP:journals/mst/KniesburgesKS14} or skip graphs~\cite{DBLP:conf/podc/JacobRSST09}.
However, our network has a better congestion than these networks, as the supervised approach allows a much more balanced distribution of these nodes.

We show how to use the supervised skip ring to obtain a self-stabilizing publish-subscribe system (see Section~\ref{sec:pub_sub_system}) in which each skip ring corresponds to a topic.
Every subscriber of a topic eventually gets all publications that have been issued so far for that topic. 
The shortcuts in the skip ring are helpful when using flooding to distribute new publications among all subscribers, since a skip ring of $n$ nodes has diameter $\log n$.

In our self-stabilizing publish-subscribe system the message overhead of the supervisor is linear in the number of topics (but not in the number of subscribers), if we use a simple generalization strategy in which each topic corresponds to one skip ring.
This, of course, decreases the applicability of our system in large-scale scenarios.
However, better scalability can be achieved by organizing topics in a hierarchical manner, or by having different supervisors for each topic.
For the latter scenario, one could make use of a self-stabilizing distributed hash table (with consistent hashing) for all supervisors, in which a sub-interval of $[0,1)$ is assigned to each supervisor.
By hashing IDs of topics  in the same manner, each supervisor is then only responsible for the topics in its sub-interval.
Since solutions for self-stabilizing distributed hash tables already exist in the literature (see e.g.~\cite{DBLP:journals/jacm/JacobRSST14}), we do not elaborate on them further in this paper.

\section{Preliminaries}\label{sec:preliminaries}
In this section we formally introduce the topology for a skip ring (Section~\ref{sec:preliminaries:son}).
As a base for our self-stabilizing protocol we introduce the \textsc{BuildRing} protocol from~\cite{DBLP:conf/ipps/KniesburgesKS12} to arrange nodes in a sorted ring (Section~\ref{sec:preliminaries:cycle}).

\subsection{Skip Ring} \label{sec:preliminaries:son}
Let $l: \mathbb{N}_0 \rightarrow \{0,1\}^*$ be a mapping with the property that $l(x) := (x_{d-1}\ldots  x_0 x_d)$ for every $x \in \mathbb{N}_0$ with binary representation $(x_d\ldots x_0)_2$ (where $d$ is minimum possible). 
Intuitively, $l$ takes the leading bit of the binary string representing the input value and moves it to the units place.
In our setting the supervisor will use $l$ to assign a (unique) label to each subscriber.
Labels are generated in the order: $0$, $1$, $01$, $11$, $001$, $011$, $101$, $111$, $0001...$.
Note that $l$ is invertible.
We call the value $l(x) \in \{0,1\}^*$ of $l$ a \emph{label}.
Denote by $|label|$ the minimum number of bits used to encode \emph{label}.
A label $y =(y_1 \ldots y_d) \in \{0,1\}^d$ may either be represented as a bit string or as a real-valued number within $[0,1)$ by evaluating the function $r:\{0,1\}^* \rightarrow [0,1)$ with $r(y) := \sum_{i=1}^d y_i/2^i$.
The function $r$ induces an ordering of all nodes in a ring, which will be used in the following to define the skip ring:

\begin{definition}[Skip Ring]\label{def:son}
	A \emph{skip ring} $SR(n)$ is a graph $G = (V, E_R \cup E_S)$ with $n$ nodes.
	$G$ is defined as follows:
	\begin{itemize}
		\item[-] Each node $v \in V$ has a unique label denoted by $label_v \in \{0,1\}^*$ with $l^{-1}(label_v) < n$.
%		\item[-] $E_{sup} := \{(v,s), (s,v)\ |\ v \in V\}$.
		\item[-] $(u,v) \in E_R \Leftrightarrow (u,v)$ are consecutive in the ordering induced by $r$. Denote the edges in $E_R$ as \emph{ring edges}.
		\item[-] $(u,v) \in E_S \Leftrightarrow (u,v)$ is part of the sorted ring w.r.t. node labels over all nodes in $K_i$, $i \in \{1,\ldots, \lceil \log n \rceil - 1\}$, where $K_i := \{w \in V\ |\ |label_w| \leq i\}$. Denote $(u,v) \in E_S$ as a \emph{shortcut on level $i$}, if $i = \max\{|label_u|, |label_v|\}$.
	\end{itemize}
\end{definition}

The label of a node $v \in V$ is independent from its unique ID $v.id$ and will be determined by the supervisor.

The intuition behind $E_R$ and $E_S$ is that we want all nodes with label of length at most $k$ to form a (bidirected) sorted ring for all $k \in \{1,\ldots,\lceil \log n \rceil\}$.
For $k = \lceil \log n \rceil$ these edges are stored in $E_R$, for $k < \lceil \log n \rceil$ they are stored in $E_S$.
Due to the way we defined the function $l$ it holds that for all $x \in \{2^d,\ldots, 2^{d+1}-1\}$ the values $r(l(x))$ are uniformly spread in between old values $r(l(y))$  with $y \in \{0,\ldots,2^d-1\}$.
This implies that the longer a node is a participant of the system, the more shortcuts it has.
This makes sense from a practical point of view, since older and thus more reliable nodes hold more connectivity responsibility in form of more shortcuts.

The decision whether two nodes are connected or not only depends on the labels of the nodes, which means that a arrival/departure of a node only affects its neighbors (see Section~\ref{sec:pub_sub_system:sub_unsub} for details).
Figure~\ref{img:ring_example} illustrates  $SR(16)$.

\begin{figure}[ht]
	\centering
 	\includegraphics[scale=0.7]{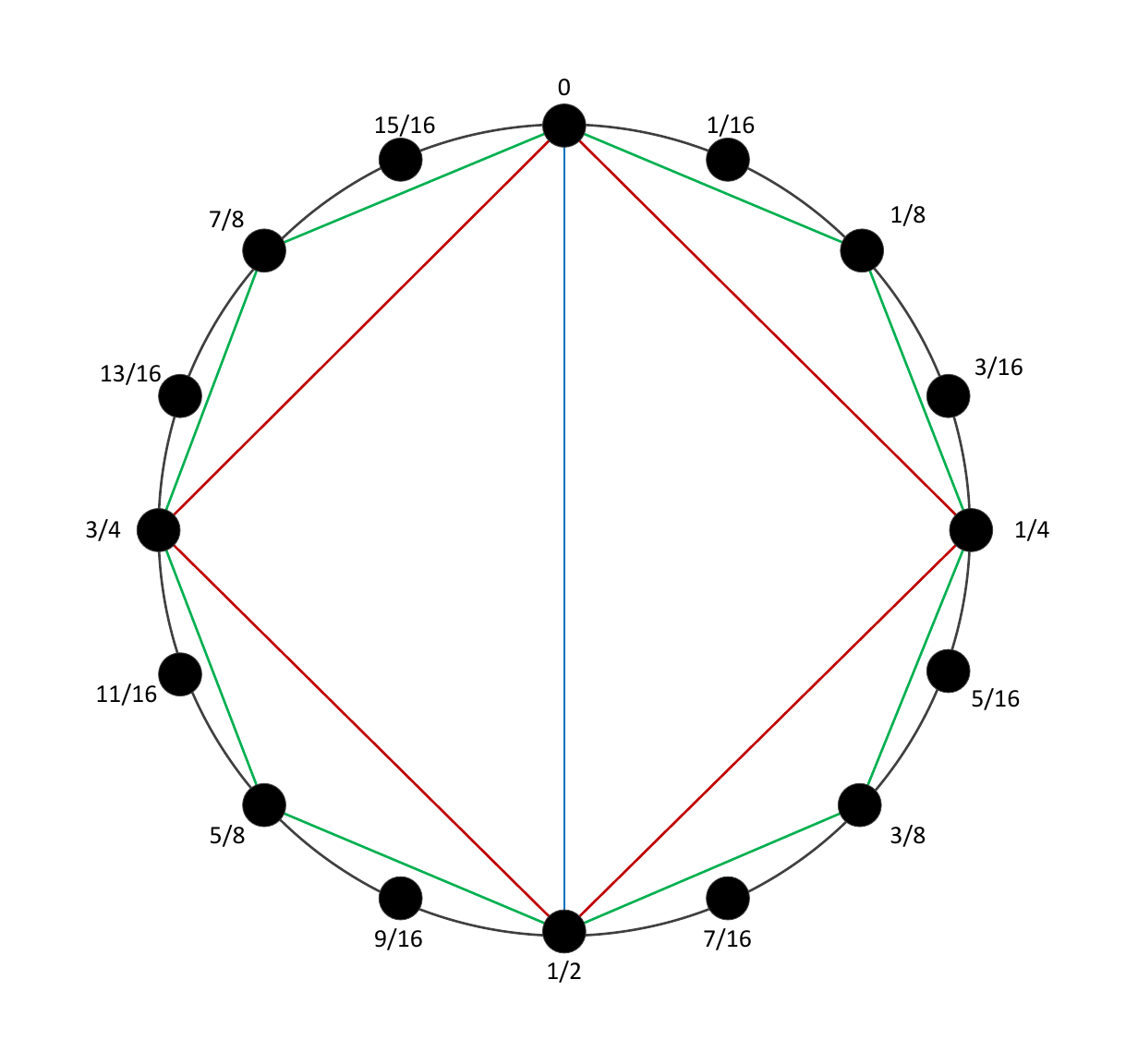}
	\caption{A skip ring consisting of 16 nodes. The triples are of the form $(x, l(x), r(l(x)))$ where $x \in \{0,\ldots,15\}$, $l(x)$ is the corresponding label and $r(l(x))$ is the real valued version of the label. Black edges are ring edges ($k = 4$), green edges are shortcuts for $k = 3$, red edges for level $k = 2$ and the blue edge is the shortcut for $k = 1$.}
	\label{img:ring_example}
\end{figure}
  
The following Lemma follows from the definition of $SR(n)$:

\begin{lemma}[Node Degree] \label{lemma:node_degree}
In a legitimate state, the degree of nodes in a skip ring is logarithmic in the worst case and constant in the average case.
\end{lemma}

\begin{proof}
For convenience, we define $k:= |label_v|$.
Node $v$ has $2$ shortcuts to nodes with label of length $k'$ for each $k' \geq k$.
Having $n$ nodes in the system, we know that $k'$ is upper bounded by $\log(n)$, which sums up the degree of $v$ to be $2 \cdot (\log n - k + 1) = \mathcal O(\log n)$.

Next, we want to compute the average degree of a node.
We count the overall number of edges in a stable system containing $n$ subscribers.
Let $f(k)$ denote the number of subscribers with label of length $k$.
We have 
\begin{align*}
 f(k) &=
  \begin{cases}
   2        & k = 1 \\
   2^{k-1}  & k > 1
  \end{cases}
\end{align*}

Recall that the maximum length of a label is equal to $\log(n)$ in a stable state.
Combining this fact with the above formula for the node degree, we get the following result for the number of edges in $|E_R \cup E_S|$:
\begin{eqnarray*}
|E_R \cup E_S| &=& \sum_{k=1}^{\log(n)} f(k) (2(\log(n)-k+1))\\
&=& 4\log(n) + \sum_{k=2}^{\log(n)} 2^{k-1} (2(\log(n)-k+1))\\
&=& 2\log(n) + \sum_{k=1}^{\log(n)} 2^k (\log(n)-k+1)\\
&=& 2n \log(n) +2n - 2 - \sum_{k=1}^{\log(n)} 2^k k\\
&=& 4n - 4\\
\end{eqnarray*}

Dividing this value by $n$ yields an upper bound of $4 = \Theta(1)$ for the average node degree. 
\end{proof}

\subsection{Self-Stabilizing Ring} \label{sec:preliminaries:cycle}
The base of our self-stabilizing protocol is the \textsc{BuildRing} protocol from~\cite{DBLP:conf/ipps/KniesburgesKS12} that organizes all nodes in a sorted ring according to their labels, using linearization~\cite{DBLP:conf/alenex/OnusRS07}: Each node $v \in V$ stores edges to its closest left and right neighbors (denoted by $v.left, v.right \in V$) according to $v$'s label (denoted by $v.label$). 
Any other nodes $u$ are delegated by $v$ to either $v.left$ or $v.right$ (depending on which node is closer to $u$).
Additionally the node with minimum label stores an edge to the node with maximum label and vice versa, such that the sorted ring is closed.
Nodes $v$ periodically \emph{introduce} themselves to their neighbors $v.left$ and $v.right$ in the sorted ring: This means that $v$ sends a message to $v.left$/$v.right$ containing a reference to itself.
This way nodes can check, if the sorted ring is in a legitimate state from their point of view or not.

In our setting nodes may assume \emph{corrupted} labels for their neighboring nodes in any nonlegal state: If node $v \in V$ has an edge to $w \in V$, then $v$ locally stores the tuple $(label_w, w)$.
While the reference to $w$ is assumed to be correct by definition at any time, $w$'s variable $w.label$ may change to a different value at some point in time.
Unfortunately, $v$ still has the old label value associated with $w$, implying that $label_w \neq w.label$.
As a consequence, we extend the \textsc{BuildRing} protocol as follows: 
Whenever a node $v \in V$ introduces itself to another node $w \in V$, then $v$ informs $w$ about the label $label_w$ that $v$ thinks is assigned to $w$.
Node $w$ then checks the label for correctness by comparing $w.label$ with $label_w$, and if $label_w \neq w.label$, $w$ sends $v$ its correct value of $w.label$.

Including the modifications mentioned above, the extended \textsc{BuildRing} protocol is still self-stabilizing:

\begin{lemma}\label{lemma:build_cycle_stabilization}
	The \textsc{BuildRing} protocol with its extension is self-stabilizing.
\end{lemma}

\begin{proof}
	In case that there are no corrupted labels, the extended \textsc{BuildRing} protocol behaves the same as the standard \textsc{BuildRing} protocol, so we refer the reader to \cite{DBLP:conf/ipps/KniesburgesKS12} to verify that the protocol is indeed self-stabilizing.
	
	We are going to show that in case there exist corrupted node labels, these will eventually vanish.
	W.l.o.g. consider the variable $u.right$ for a node $u$.
	We define the \emph{trace} of $u.right$ as the chain of values $((label_{v_1}, v_1)$, $ (label_{v_2}, v_2), \ldots)$ that are allocated to $u.right$ while the system stabilizes.
	By definition of the standard \textsc{BuildRing} protocol, labels stored in the trace for $u.right$ are monotonically decreasing.
	Furthermore, the trace is finite, since the number of labels (corrupted or correct) is finite.
	Let $(label_{v_k}, v_k)$ be the last node of this trace, i.e., eventually it holds $u.right = (label_{v_k}, v_k)$.
	Then $u$ will introduce itself in its \textsc{Timeout} method to $v_k$ by sending a message $M$ storing $u$ itself and $label_{v_k}$ to $v_k$.
	Upon receiving $M$, $v_k$ is able to check if $label_{v_k} = v_k.label$ and send a reply storing the (correct) label $v_k.label$ to $u$ in case $label_{v_k} \neq v_k.label$, s.t. $u$ corrects its label $label_{v_k}$.
	This implies that the number of corrupted labels is reduced, but there may be a new trace generated for $u.right$.
	But since the number of corrupted node labels is finite and is not duplicating, eventually, the overall number of corrupted node labels will reduce to $0$.
\end{proof}

\section{Self-stabilizing Supervised Skip Ring}\label{sec:protocol}
In this section we first extend the skip ring topology by introducing a supervisor.
The description of our \textsc{BuildSR} is then split into two sub-protocols: One sub-protocol is executed by the supervisor, the other one is executed by every other node.
Adapting publish-subscribe terminology, we denote a node $v \in V$ as \emph{subscriber} for the rest of the paper.

Recall that every subscriber $v$ is assumed to know the supervisor $s$, and this information is read-only, so the graph $G$ always contains edges $(v,s)$.
The assumption of having such a supervisor is not far-fetched, because even pure peer-to-peer systems need a common gateway that acts as an entrance point for peers.

Our goal in this section is to construct a self-stabilizing protocol in which subscribers form a skip ring with the help of the supervisor, starting from any initial state.
The extension to a self-stabilizing publish-subscribe system is then described in Section~\ref{sec:pub_sub_system}.

\subsection{Supervisor Protocol} \label{sec:supervisor}
The first part of the \textsc{BuildSR} protocol is executed by the supervisor.
The supervisor maintains the following variables:
\begin{description}
	\item[-] A $database \subset \{0,1\}^* \times V$ containing subscribers and their corresponding labels. Denote $n := |database|$.
	\item[-] A variable $next \in \mathbb{N}$ that is used to notify subscribers in a round-robin fashion.
\end{description}

In the supervisor's \textsc{Timeout} method, the supervisor chooses a subscriber $v$ in a round-robin fashion (using the variable $next$) from its $database$.
Then the supervisor sends a message to $v$ containing $v$'s label $label_v$, as well as the correct values for $v$'s predecessor $pred_v$ and successor $succ_v$ according to the $database$.
We call such a triple $(pred_v, label_v, succ_v)$ the \emph{configuration} for $v$.

In addition to the above action, the supervisor has to check the integrity of its $database$: We say that the $database$ of $s$ is \emph{corrupted}, if at least one of the following conditions hold:
	\begin{itemize}
		\item[(i)] There exists a tuple $(label,v) \in database$ with $v = \perp$. (There exists a tuple without any subscriber)
		\item[(ii)] There exist entries $(label_1, v_1), (label_2, v_2) \in database$ with $label_1 \neq label_2$ and $v_1 = v_2$. (There exist multiple tuples storing the same subscriber)
		\item[(iii)] There exists $i \in \{0,\ldots,n-1\}$, s.t. for all $(label_v,v) \in database$ it holds $label_v \neq l(i)$. (There are labels missing)
		\item[(iv)] There exists $i \geq n$, s.t. there is a tuple $(label,v) \in database$ with $label = l(i)$. (There exists a tuple with an incorrect label)
	\end{itemize}

All of these cases may occur in initial states.
Note that when using a hashmap for the $database$, we do not need to check explicitly whether there are multiple tuples with the same label or whether there are tuples with the label set to $\perp$.
We perform the following actions to tackle the above 4 cases:

\begin{itemize}
	\item[(i)] Upon detecting a tuple $(label, \perp) \in database$, we simply remove it from the $database$.
	\item[(ii)] Whenever a subscriber $v$ wants to unsubscribe or request its configuration, the supervisor first searches the database for all tuples $(l,w)$ with $w = v$.
	It then removes all duplicates except the tuple with lowest label, guaranteeing that $v$ is associated with no more than one label.
	\item[(iii)] In \textsc{Timeout} the supervisor checks for all $i \in \{0,\ldots,n-1\}$ if there is a tuple $(l(i), v)$ stored in the $database$.
	If not, then the supervisor takes the tuple $(l(j), w) \in database$ with maximum $j \in \mathbb{N},\ j > i$ and replaces its label with $l(i)$.
	\item[(iv)] It is easy to see that the action for (iii) also solves case (iv).
\end{itemize}

Observe that all of these actions are performed locally by the supervisor, i.e., they generate no messages. 
Therefore we assume that the $database$ of the supervisor is always in a non-corrupted state from this point on.

\subsection{Subscriber Protocol}\label{sec:subscriber}
In this section we discuss the part of the \textsc{BuildSR} protocol that is executed by each subscriber.
First, we present the variables needed for a subscriber.
Note that we intentionally omit the reference to the supervisor $s$ here, since $s$ is assumed to be hard-coded.
A subscriber $v \in V$ stores the following variables:
	\begin{itemize}
		\item[-] $v.label \in \{0,1\}^* \cup \{\perp\}$: The unique label of $v$ or $\perp$ if $v$ has not received a label yet.
		\item[-] $v.left, v.right, v.ring \in \{0,1\}^* \times (V \cup \{\perp\})$: Left and right neighbor in the ring as well as the cyclic connection in case $r(v.label)$ is minimal/maximal.
		\item[-] $v.shortcuts \subset \{0,1\}^* \times (V \cup \{\perp\})$: All of $v$'s shortcuts.
	\end{itemize}

For the rest of the protocol description, we use $v.left$ and $v.right$ to indicate $v$'s left (resp. right) neighbor in the ring even if the left (resp. right) neighbor is stored in $v.ring$ instead of $v.left$ (resp. $v.right$).
We also may refer to the variables $v.left, v.right, v.ring$ as $v$'s \emph{direct ring neighbors}.
Recall that each subscriber executes the extended \textsc{BuildRing} protocol from Section~\ref{sec:preliminaries:cycle}.

\subsubsection{Receiving correct Labels} \label{sec:subscriber:label_protocol}
For now we focus on the ring edges only.
Our first goal is to guarantee that every subscriber $v$ eventually stores its correct label in $v.label$.

Recall that we have periodic communication from the supervisor to the subscribers, i.e., the supervisor periodically sends out the configurations to all subscribers stored in its $database$.
This action alone does not suffice in order to make sure that every subscriber eventually stores its correct label, since in initial states the $database$ may be empty and subscriber labels may store arbitrary values.
Thus, we also need periodic communication from subscribers to the supervisor.
The challenge here is to not overload the supervisor with requests in legitimate states of the system.
Each subscriber $v$ periodically executes the following actions:

\begin{itemize}
	\item[(i)] If $v.label = \perp$, then $v$ asks the supervisor to integrate $v$ into the $database$ and send $v$ its correct configuration.
	\item[(ii)] If $v.label \neq \perp$, then, with probability $1/(2^k\cdot k^2)$, $v$ asks the supervisor for its correct configuration, where $k = |v.label|$.
\end{itemize}

Action (ii) is dedicated to handle subscribers that have incorrect labels or already store a label, but are not known to the supervisor.
Upon receiving a configuration request from a subscriber $v$, the supervisor integrates $v$ into the $database$ (if it is not already contained in the $database$) and sends $v$ its configuration and thus its correct label.

We still need some further actions to tackle special initial states: Imagine a subscriber $v$ having a label such that the probability mentioned in (ii) becomes negligible.
In case $v$ is not contained in the supervisor's $database$ yet, $v$ will send a configuration request to the supervisor with only very low probability.
The following action is able to solve this problem under the assumption that there exists a subscriber $w$ that is already contained in the supervisor's $database$ and has $v$ stored as one of its direct ring neighbors.

\begin{itemize}
	\item[(iii)] W.l.o.g. let $w.left = (label_v, v)$. If $w$ receives a configuration from the supervisor and $pred \neq v$, it checks whether $w.left$ is \emph{closer} than the left ring neighbor $(label_p, pred) \in \{0,1\} \times V$ proposed by the configuration, i.e., $|r(label_v) - r(w.label)| \leq  |r(label_p) - r(w.label)|$.
	In case this holds, $w$ requests the supervisor to send the correct configuration to $w.left$.
\end{itemize}

The assumption for action (iii) may not hold in all initial states, i.e., there is a connected component in which all subscribers have stored labels such that the probability mentioned in (ii) becomes negligible.
Note that actions (i)-(iii) suffice to show convergence in theory.
In order to improve the time it takes the network to converge, we introduce one last periodic action:

\begin{itemize}
	\item[(iv)] Subscriber $v$ periodically requests its configuration with probability $1/2$ from the supervisor if it determines, based only on its local information, that its label is minimal.
\end{itemize}

We now sketch why eventually all subscribers in a connected component $C$ get their correct label.
This is obviously the case when all subscribers in $C$ are stored in the supervisor's $database$ as the supervisor will then periodically hand out the correct labels in a round-robin fashion.
Denote a subscriber that is already stored in the supervisor's $database$ as \emph{recorded}.
Action (iv) guarantees that we quickly have at least one recorded subscriber in a connected component $C$.
Assume that $C$ still contains non-recorded subscribers.
As long as the supervisor is able to introduce new recorded subscribers to recorded subscribers in $C$, $C$'s size grows, but since the number of subscribers is finite, $C$ will eventually become static.
For such a static connected component $C$ we know that due to \textsc{BuildRing}, subscribers in $C$ eventually form a sorted ring.
Then there exists at least one ring edge from a recorded subscriber $v \in C$ to a non-recorded subscriber $w \in C$.
Furthermore, $v$'s correct ring neighbor $w'$ indicated by its configuration has to be further away from $v$ than $w$ (so for instance $r(v.label) < r(w.label) < r(w'.label)$ if we consider the right neighbor of $v$).
This holds because no new subscriber can be introduced to a recorded subscriber in $C$.
Once $v$ receives its configuration from the supervisor, it triggers action (iii) and requests the configuration for $w$, leading to $w$ being inserted into the supervisor's $database$ and thus reducing the number of non-recorded subscribers in $C$ by one.
This inductively implies that eventually all subscribers in $C$ are recorded.
We now want to bound the expected number of requests that are periodically sent out to the supervisor when the system is in a legitimate state.
For the next lemma, denote a \emph{timeout interval} as the time in which every subscriber has called its \textsc{Timeout} method exactly once.

\begin{theorem} \label{lemma:supscriber_to_supervisor_communication}
Consider a supervised skip ring with $n$ subscribers in a legitimate state.
The expected number of configuration requests sent out by all subscribers is less than $1$ in each timeout interval.
\end{theorem}

\begin{proof}
Since the system has $n$ subscribers, the maximum length of a subscriber's label is equal to $\log(n)$ in a legitimate state.
In a legitimate state, only the second action (ii) is executed by subscribers, as all subscribers have stored their correct configuration.
Thus, requests are only sent from a subscriber $v$ to the supervisor with probability based on $v$'s label length $|v.label|$.
The number of subscribers with label of length $k$ is equal to $2^{k-1}$ and the probability that a subscriber with label of length $k$ contacts the supervisor in its \textsc{Timeout} procedure is equal to $1/(2^k\cdot k^2)$.
It follows that the expected number of configuration requests sent out by subscribers with label of length $k$ is equal to $\sum_{i=1}^{2^{k-1}} 1/(2^k \cdot k^2) = 1/(2k^2)$.
In summary, the expected number of configuration requests that are sent out by all subscribers is equal to $\sum_{k=1}^{\log(n)}1/(2k^2) < 1$.
\end{proof}

\subsubsection{Maintaining Shortcuts}\label{sec:subscriber:shortcuts}
In this section we describe how subscribers establish and maintain shortcut edges.
Recall that we have shortcuts on levels $k = \{1,\ldots,\lceil \log n \rceil\}$, where $k = \lceil \log n \rceil$ represents the ring edges that are already established.
A subscriber $v$ with label of length $k = |v.label|$ has exactly 2 shortcuts on each level in $\{k,\ldots,\lceil \log n \rceil\}$ in a legitimate state.

We first describe how a subscriber is able to compute all its shortcut labels locally, based only on the information of its left and right direct ring neighbors.
The following approach only computes the respective labels in $[0,1)$ that a node should have shortcuts to, but not the subscribers that are associated with these labels.
The idea is the following: In general, a subscriber $v \in V$ has only shortcuts to other subscribers that lie on the same semicircle as $v$, i.e., either the semicircle of subscribers within the interval $[0,1/2]$ or the semicircle of subscribers within the interval $[1/2, 1]$ (where the $1$ is represented by the subscriber with label $0$).
Consider a subscriber $v$ with $r(v.label) \in [0,1)$ and its two ring neighbors $w, u$ such that $v.left = (label_w, w)$ and $v.right = (label_u, u)$.
If $v$ recognizes that $|v.label| < |label_w|$, then $v$ knows that it has to have a shortcut with label $s$ and $r(s) = 2\cdot r(label_w) - r(v.label)$, because node $w$ was previously inserted between subscribers with labels $s$ and $v.label$.
After this, $v$ can apply this method recursively, i.e., it checks for the computed label $s$ if $|v.label| < |s|$ until it reaches a label of less or equal length.
This same procedure is applied analogously for $v.right$.

As an example, recall the (stable) ring from Figure~\ref{img:ring_example}.
Suppose we want to compute all shortcut labels for the subscriber with (real-valued) label $1/4$, based only on the labels of its direct ring neighbors, which are $3/16$ and $5/16$.
We know that the label $3/16$ has length $4$, which is greater than the length of label $1/4$, which is $2$.
Thus, we get a shortcut $s_1$ for $1/4$ with label $2\cdot(3/16) - 1/4 = 1/8$.
The label $1/8$ has length $3$, which is still greater than $2$.
Hence we compute a shortcut $s_2$ with label $2\cdot(1/8) - 1/4 = 0$.
Finally we know that the length of label $0$ is $1$, which is smaller than $2$, which terminates the algorithm.
The computation of shortcut labels to $3/8$ and $1/2$ works analogously.

We are now ready to define the self-stabilizing protocol that establishes and maintains shortcuts for all subscribers.
Consider a subscriber $v$ with label length $|v.label| = k$.
On \textsc{Timeout}, $v$ checks if $v.shortcuts$ contains subscribers $(label_u, u), (label_w, w)$ on level $k$.
If that is the case, then $v$ introduces $u$ to $w$, by sending a message to $w$ containing the reference of $u$ as well as $u$'s label $label_u$.
Also, $v$ introduces $w$ to $u$ in the same manner.
Note that for $|v.label| = \lceil \log n \rceil$, $v$ has to consider its two ring neighbors instead of $v.shortcuts$.
On receipt of such an introduction message consisting of the pair $(label_w, w)$, $u$ checks if it has a shortcut $(label_{w'}, w')$ with $label_{w'} = label_w$.
If that is the case, then $u$ replaces the existing node reference $w'$ by $w$ and, if $w' \neq w$, forwards the reference of $w'$ on the sorted ring via the \textsc{BuildRing} protocol.
This way it is guaranteed that shortcuts are established in a bottom-up fashion.

\subsection{Handling Subscriber Failures}\label{sec:subscriber:failures}
We now consider the case that subscribers $v \in V$ are allowed to \emph{crash} without warning.
In this case the address $v.id$ ceases to exist.
Consequently, even though nodes may still send messages to $v$, these messages do not invoke any action on $v$.
Note that we do not consider supervisor failure, since it is assumed to be hard-coded.
The challenge here is to restore the system to a correct supervised skip ring that does not contain $v$, i.e., we need to exclude $v$ from the system.
In pure peer-to-peer systems this scenario is a problem, since we have to maintain failure detectors~\cite{DBLP:journals/jacm/ChandraT96} at each node in order to be able to determine if some neighboring node has crashed.
This leads to an increased overhead in the complete system.
However, in our setting it suffices to establish only one single failure detector at the supervisor, because we only need to make sure that the $database$ will eventually contain the correct data.
Consequently, if the supervisor notices that subscriber $v$ has crashed, it just has to remove $v$ from its $database$.
By periodically executing the actions for restoring a corrupted $database$ we know that the $database$ will eventually contain the correct data.

\section{Self-Stabilizing Publish-Subscribe System}\label{sec:pub_sub_system}
In this section we show how to use our \textsc{BuildSR} protocol as a self-stabilizing publish-subscribe system.
We start by discussing some general modifications and then describe the operations \emph{subscribe}, \emph{unsubscribe} and \emph{publish}.

Let $\mathcal T \subset \mathbb{N}$ be the set of available topics that one may subscribe to.
To construct a publish-subscribe system out of our self-stabilizing supervised overlay network, we basically run a \textsc{BuildSR} protocol for each topic $t \in \mathcal T$ at the supervisor.
Thus, the supervisor has to extend its $database$ to be in $(\{0,1\}^*)^{|\mathcal T|} \times V$.
From here on, we assume that each message contains the topic it refers to, such that the receiver of such a message can match it to the respective \textsc{BuildSR} protocol.
Once a subscriber wants to subscribe to some topic $t \in \mathcal T$, it starts running a new \textsc{BuildSR} protocol for topic $t$.
Upon unsubscribing, the subscriber may remove the respective \textsc{BuildSR} protocol, once it gets the permission from the supervisor, implying that the supervisor has removed the subscriber from its $database$.
By assigning the topic number to each message that is sent out, we can identify the appropriate protocol at the receiver.
For convenience, we still consider only one supervised skip ring for the rest of the paper.

\subsection{Subscribe/Unsubscribe} \label{sec:pub_sub_system:sub_unsub}
When processing a $subscribe(v)$ operation, the supervisor executes the following actions (denote by $n$ the number of nodes in the $database$ before the subscribe/unsubscribe request):

\begin{enumerate}
	\item Insert $(l(n), v)$ into the $database$.
	\item Send $v$ its correct configuration $(pred_v, l(n), succ_v)$.
\end{enumerate}

The correctness of subscribe follows immediately, since our protocol is self-stabilizing.
Note that the supervisor can easily extract the tuples $pred_v$ and $succ_v$ from the $database$, since all tuples are sorted based on the value of their labels.
Our approach has the advantage that it spreads multiple sequential subscribe operations through the skip ring, meaning that a pre-existing subscriber is involved (i.e., it has to change its configuration) only for two consecutive subscribe operations. Afterwards its configuration remains untouched until the number of subscribers has doubled.
This is due to the definition of the label function $l$.
As an example consider the skip ring $SR(16)$ from Figure~\ref{img:ring_example} and assume that there are $16$ new subscribers that want to join.
Then these new subscribers are inserted in between consecutive pairs of old subscribers on the ring, as they receive (real-valued) labels $1/32$, $3/32$, $5/32, \ldots$.

When processing an $unsubscribe(v)$ operation, the supervisor executes the following actions:

\begin{enumerate}
	\item Remove $(label_v, v)$ from the $database$.
	\item Get the tuple $(label_w, w)$ with $label_w = l(n-1)$ from the $database$ and replace $label_w$ with $v$'s label $label_v$ in the $database$.
	\item Send $w$ its new configuration $(pred_v, label_v, succ_v)$.
	\item Inform $v$ that it is granted permission to delete all its connections to other subscribers.
\end{enumerate}

After both subscribers have received their correct label from the supervisor, the ring will stabilize itself.
Note that the supervisor's $database$ is already in a legitimate state after the initial $subscribe$ (resp. $unsubscribe$) message has been processed by the supervisor.
Therefore, the supervisor does not rely on additional information from subscribers to stabilize its $database$.
The following lemma states the correctness of \emph{unsubscribe}.

\begin{lemma} \label{lemma:unsubscribe}
	After a subscriber $v$ has sent an $unsubscribe(v)$ request to the supervisor, $v$ eventually gets disconnected from the graph induced by $E_R \cup E_S$.
\end{lemma}

\begin{proof}
	Assume that subscriber $v \in V$ sends an $unsubscribe(v)$ request to the supervisor.
	By definition of the unsubscribe protocol, the supervisor removes $v$ from its $database$ and sends $v$ its configuration that is $(\perp, \perp, \perp)$.
	Hence, $v$ sets $v.label = \perp$ and answers all incoming introduction messages from other subscribers with the request to delete the connection to $v$.
	By definition of \textsc{BuildSR}, every subscriber $u$ that has an edge $(u, v) \in E_R \cup E_S$ will eventually introduce itself to $v$, leading to $v$ eventually getting disconnected from the graph induced by $E_R \cup E_S$.
	This proves the lemma.
\end{proof}

It follows from the above descriptions that the supervisor only has to send out a constant number of messages per subscribe/unsubscribe request:

\begin{theorem}
	In a legitimate state, the message overhead of the supervisor and subscribers is constant for subscribe/unsubscribe operations.
\end{theorem}

\subsection{Publish}\label{sec:pub_sub_system:pub}
In the following paragraphs we extend our protocol to be able to provide \emph{publish} operations in a self-stabilizing manner.
Note that the presented approach is used \emph{only} to realize a self-stabilizing publication-dissemination-approach.
There exist dedicated protocols (e.g. flooding, see Section~\ref{app:publish:flooding}) that realize a more efficient distribution of publications among the subscribers.
A self-stabilizing protocol for publications is able to correct eventual mistakes that occurred in the flooding approach.
For storing publications at each subscriber, we use an extended version of a Patricia trie~\cite{DBLP:journals/jacm/Morrison68} to effectively determine missing publications at subscribers.
We first define the Patricia trie and later on present a protocol that is able to merge all publications in all Patricia tries.
This results in each subscriber storing all publications.

A \emph{trie} is a search tree with node set $T$ over the alphabet $\Sigma = \{0,1\}$. Every edge is associated with a label $c \in \Sigma$.
Additionally, every key $x \in \Sigma^k$ that has been inserted into the trie can be reached from the root of the trie by following the unique path of length $k$ whose concatenated edge labels result in $x$.

A \emph{Patricia trie} is a compressed trie in which all chains (i.e., maximal sequences of nodes of degree 1) are merged into a single edge whose label is equal to the concatenation of the labels of the merged trie edges.
We store a Patricia trie at each subscriber $v \in V$, denoted by $v.T$.
Each leaf node in a Patricia trie stores a publication $p \in \mathcal P^*$, where $\mathcal P = \{0,1\}$ is the alphabet for publications.
Note that each inner node $t \in T$ of a Patricia trie $T$ has exactly 2 child nodes denoted by $c_1(t), c_2(t) \in T$.
Furthermore, we want to assign a label to each node: The label $t.label \in \Sigma^k$ of an inner node $t \in T$ is defined as the longest common prefix of the labels of $t$'s child nodes (with $\perp$ being the empty word).
If $t$ is a leaf node storing a publication $p \in \mathcal P^*$, we define $t$'s label to be the unique key generated by the collision-resistant hash function $\bar{h}_m: \mathbb{N} \times \mathcal P^* \rightarrow \{0,1\}^m$, where a pair $(v.id, p) \in \mathbb{N} \times \mathcal P^*$ contains the unique ID of the subscriber $v \in V$ that generated the publication $p$.
Note that the constant $m \in \mathbb{N}$ and the hash function $\bar{h}_m$ are known to all subscribers, ensuring that every label for a publication has the same length.

In addition to node labels, we let nodes store (unique) hash values: We use another collision-resistant hash function $h: \{0,1\}^* \rightarrow \{0,1\}^*$ and define the hash value $t.hash$ of a leaf node $t$ as $h(t.label)$.
If $t$ is an inner node, then $t.hash$ is defined as the hash of the concatenation of the hashes of $t$'s child nodes, i.e., $t.hash := h(h(c_1(t).label) \circ h(c_2(t).label))$.
This approach is similar to a Merkle-Hash Tree (MHT)~\cite{DBLP:conf/crypto/Merkle87}, which also hashes data using a collision-resistant hash function and building a tree on these hashes.
However, our approach does not require one-way hash functions, which is a standard assumption in MHTs, because we do not require our scheme to be cryptographically secure.

If a subscriber $v \in V$ wants to publish a message $p \in \mathcal P^*$ over the ring, $v$ just inserts $p$ into its own Patricia trie.
The publication $p$ is then spread among all subscribers of the ring by the following protocol, executed at each subscriber $v \in V$: Subscriber $v$ periodically sends a request \Call{CheckTrie}{$v$, $r_v$} to one of its ring neighbors (chosen randomly) containing $v$ itself and the root node $r_v \in v.T$ of $v$'s Patricia trie.
Note that sending an arbitrary node $t \in v.T$ along a message \textsc{CheckTrie} means that we only store $t.label$ and $t.hash$ in the request while ignoring $t$'s outgoing edges.
Upon receiving a request \Call{CheckTrie}{$v$, $t_v$} with $t_v \in v.T$, a subscriber $u \in V$ does the following: It searches for the node $t_u \in u.T$ with label $t_u.label = t_v.label$ and checks if $t_u.hash = t_v.hash$.
The following three cases may happen:

\begin{itemize}
	\item[(i)] $t_u.hash = t_v.hash$: Then we know that the set of publications stored in the subtrie of $u.T$ with root node $t_u$ are the same as the set of publications stored in the subtrie of $v.T$ with root node $t_v$.
	Subscriber $u$ does not send any response to $v$ in this case.
	\item[(ii)] $t_u.hash \neq t_v.hash$: Then the contents of the subtries with roots $t_u, t_v$ differ in at least one publication.
	In order to detect the exact location, where both Patricia tries differ, $u$ responds to $v$ by sending a request \Call{CheckTrie}{$u$, $c_1(t_u)$, $c_2(t_u)$} to $v$, which is handled by $v$ as two separate \textsc{CheckTrie} requests \Call{CheckTrie}{$u$, $c_1(t_u)$} and \Call{CheckTrie}{$u$, $c_2(t_u)$}.
	\item[(iii)] $t_v$ does not exist in $u.T$: Then $v.T$ contains publications that do not exist in $u.T$.
	Subscriber $u$ is able to compute the label prefix of those missing publications: First, $u$ searches for the node $c \in u.T$ with label prefix $t_v.label$ and $|c.label|$ minimal, i.e., $c.label = t_v.label \circ b_1 \circ \ldots \circ b_k$ with $b_1, \ldots,  b_k \in \{0,1\}$ and $|c.label| = |t_v.label| + k$ minimal.
	If such a node $c$ exists, then $u.T$ may contain at least all publications with label  prefix $c.label$.
	Furthermore, $u$ knows that all publications with label prefix $t_v.label \circ (1-b_1)$ are missing in its Patricia trie.
	As a consequence, $u$ requests $v$ to continue checking the subtrie with root node of label $c.label$ and to deliver all publications with label prefix $p = t_v.label \circ (1-b_1)$ to $u$.
	It does so by sending a \Call{CheckAndPublish}{$u$, $c$, $p$} request to $v$, where $v$ internally calls \Call{CheckTrie}{$u$, $c$} and, in addition, delivers all publications with label prefix $p$ to $u$.
	In case that a node $c$ as described above cannot be found in $u.T$, $u$ just requests $v$ to deliver all publications with prefix $t_v.label$ to $u$, since that subtrie is missing in $u.T$.
\end{itemize}

With this approach, only those publications are sent out that are assumed to be missing at the receiver.

As an example consider two subscribers $u, v \in V$ with Patricia tries as shown in Figure~\ref{img:partricia_trie_ss_example}.
Note that $P_4$ is missing in $v.T$.
We describe how $v$ will eventually receive $P_4$.

\begin{figure*}[ht]
	\centering
 	\includegraphics[scale=0.9]{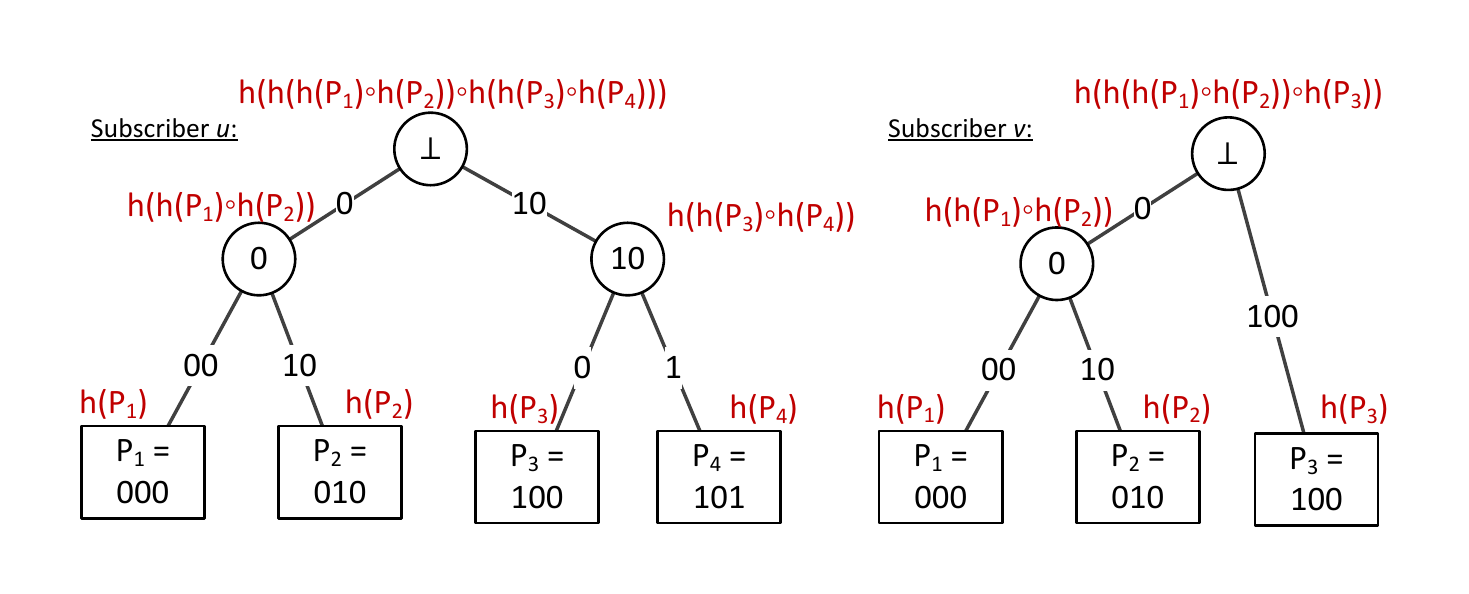}
	\caption{Example Patricia tries $u.T$ and $v.T$ for two subscribers $u, v \in V$.}
	\label{img:partricia_trie_ss_example}
\end{figure*}

First assume that $u$ sends out a \Call{CheckTrie}{$u$, $r_u$} message to $v$ in its \textsc{Timeout} method with $r_u$ being the root node of $u$'s Patricia trie.
Subscriber $v$ then compares the hash $r_u.hash$ with the hash of its root node, which is not equal.
Thus, $v$ sends a message \Call{CheckTrie}{$v$, $(0, h(h(P_1)\circ h(P_2))), (100, h(P_3))$} to $u$, which forces $u$ to compare the hashes the nodes with labels $0$ resp. $100$ to the hashes $h(h(P_1)\circ h(P_2))$ resp. $h(P_3)$.
Both comparisons result in the hashes being equal, which ends the chain of messages at subscriber $u$.

Now assume, that $v$ sends out a \Call{CheckTrie}{$v$, $r_v$} message to $u$ in its \textsc{Timeout} method with $r_v$ being the root node of $v$'s Patricia trie.
Subscriber $u$ compares the $r_v.hash$ with $r_u.hash$ and spots a difference.
Thus, it sends a message \Call{CheckTrie}{$u$, $(0, h(h(P_1)\circ h(P_2))), (10, h(h(P_3)\circ h(P_4)))$} to $v$.
For the node with label $0$ this results in both hashes being equal, but then $v$ cannot find a node with label $10$ in its Patricia trie, which is why $v$ sends a message \Call{CheckAndPublish}{$v$, $(100, h(P_3))$, $p=101$} to $u$.
Note that the node with label $100$ is the node with label of minimum length for which $10$ is a prefix.
Thus, $p = (10 \circ (1-0)) = 101$.
The \textsc{CheckAndPublish} request forces $u$ to compare the hashes of its node with label $100$ to the given hash $h(P_3)$, which results in both hashes being equal.
Furthermore, $u$ sends all publications with labels of prefix $101$ to $v$, which is only the publication $P_4$.
After $v$ has inserted $P_4$ in its Patricia trie, both tries are equal, resulting in two equal root hashes.

The example shows that it is important at which subscriber the initial \textsc{CheckTrie} request is started.

\subsection{Flooding} \label{app:publish:flooding}
As an extension to the above approach, we can make use of shortcuts to spread new publications over the ring: Whenever a subscriber $u \in V$ generates a new publication $p$, $u$ inserts $p$ into $u.T$ and broadcasts $p$ over the ring, by sending a \Call{PublishNew}{$p$} message to all of its neighbors $v$ with $(u,v) \in E_R \cup E_S$.
Upon receiving such a \Call{PublishNew}{$p$} message, a subscriber $v \in V$ checks if $p$ is already stored in $v.T$.
If not, then $v$ inserts $p$ into $v.T$ and continues to broadcast $p$ by forwarding the \textsc{PublishNew} message to its neighbors.
In case that $p$ is already stored in $v.T$, $v$ just drops the message.
By applying this flooding approach on top of the self-stabilizing publish protocol, we can achieve faster delivery of new publications in practice (recall that the skip ring has a diameter of $\log n$).
Still, if a new subscriber joins a topic, it has to rely on the core \textsc{BuildSR} protocol to receive all publications.
Furthermore, note that we do not rely on flooding to show convergence of publications.

\section{Analysis}\label{sec:analysis}
In this section we show that \textsc{BuildSR} is self-stabilizing according to Definition~\ref{def:self_stabilization}.
We also show that eventually all subscribers are storing all publications in their respective Patricia tries.
The combination of the first two theorems yields that \textsc{BuildSR} is self-stabilizing:

\begin{theorem}[Network Convergence]\label{theorem:network_convergence}
	Given any initially weakly connected graph $G = (V, E_R \cup E_S)$ with $n$ nodes, \textsc{BuildSR} transforms $G$ into a skip ring $SR(n)$.
\end{theorem}

\begin{proof}
First of all, note that eventually all corrupted messages are received.
Furthermore, a corrupted message cannot trigger an infinite chain of corrupted messages, i.e., eventually the false information is either corrected or received but not spread anymore.
We assume this fact for the rest of the proof.

We start with the supervisor and prove the following lemma:

\begin{lemma}[Supervisor Validity] \label{lemma:database_convergence}
	Eventually the supervisor's $database$ will not be corrupted anymore.
\end{lemma}

\begin{proof}
	We show for every condition that may occur in a corrupted $database$ that the condition does not occur anymore at some point in time.
	\begin{itemize}
		\item[(i)] Assume that there exists a tuple $(label,v) \in database$ with $v = \perp$.
		Then this tuple is simply deleted from the database (see Algorithm~\ref{algorithm:supervisor}, line ~\ref{algorithm:database_integrity:algline:checkLabels_remove_empty_tuple}).
		\item[(ii)] Assume that there are entries $(label_1, v_1), (label_2, v_2) \in database$, $label_1 \neq label_2$ and $v_1 = v_2$.
		At some point in time the supervisor wants to send the configuration to $v_1$  (Algorithm~\ref{algorithm:supervisor}, line~\ref{algorithm:supervisor:algline:timeout_sendConfig}), which leads to the supervisor calling \Call{CheckMultipleCopies}{$v_1$}.
		Therefore, the supervisor is able to detect the entries $(label_1, v_1)$, $(label_2, v_2)$.
		W.l.o.g. assume that $(label_1, v_1)$ is detected before $(label_2, v_2)$ in the for-loop of \textsc{CheckMultipleCopies}.
		Then the supervisor just removes the redundant tuple $(label_2, v_2)$ from the $database$, which resolves this condition.
		\item[(iii)] Assume there exists $i \in \{0,\ldots,n-1\}$, such that for all $(label_v,v) \in database$ it holds $label_v \neq l(i)$.
		In this case, the supervisor is able to detect this corruptness:
		It asks for the entry with label $l(i)$ in the $database$ (Algorithm~\ref{algorithm:supervisor}, line~\ref{algorithm:database_integrity:algline:checkLabels_getli}).
		Then the supervisor proceeds as follows:
		It replaces the label of the entry $(l(j), w)$ with maximum $j > i$ with the label $l(i)$ (see Algorithm~\ref{algorithm:supervisor}, line ~\ref{algorithm:database_integrity:algline:checkLabels_replace}).
		Note that $(l(j), w)$ has to exist in this case, since otherwise $i \geq |database| = n$, which is a contradiction.
		\item[(iv)] Finally, assume that there exists $i \geq n$, s.t. there is a tuple $(label,v) \in database$ with $label = l(i)$.
		Then there exists a number $j \in \{0,\ldots,n-1\}$, for which there is no tuple $(l(j), w) \in database$, so eventually $(label,v)$ gets its label changed to $(l(j), v)$.
	\end{itemize}
	In the end we showed that all 4 conditions that may occur in a corrupted $database$ do not occur in the $database$ from some point in time on, which concludes the proof.
\end{proof}

Using this lemma, we are ready to prove the convergence of the supervisor:

\begin{lemma} [Supervisor Convergence]\label{lemma:supervisor_convergence}
	After the supervisor's $database$ has reached a non-corrupted state, all subscribers $v \in V$ will eventually become recorded.
\end{lemma}

\begin{proof}
Recall that we call a subscriber \emph{recorded}, if it is stored in the supervisors's $database$.
Note that the supervisor does not remove subscribers from its (non-corrupted) $database$ unless it gets the request to do so, which is not the case as we do not consider \textsc{unsubscribe} operations here.

Let $v \in V$ be a non-recorded subscriber.
Then either $v.label = \perp$ or $v.label \neq \perp$.
Recall the actions described in Section~\ref{sec:subscriber:label_protocol}.
If $v.label = \perp$, then $v$ requests its configuration from the supervisor via action (i) and becomes recorded.

Now let $v.label \neq \perp$.
If $v$ does not have connections to other subscribers (i.e., $v.left = v.right = \perp$), then $v$ requests its configuration from the supervisor via action (iv), as $v$ thinks that it is the subscriber with minimal label, so $v$ becomes recorded.

It remains to consider the general case, where we are given a connected component $C$ of subscribers.
Action (iv) guarantees that $C$ quickly contains at least one subscriber that is recorded.
As long as the supervisor is able to introduce new recorded subscribers to recorded subscribers in $C$, $C$'s size grows, but since the number of subscribers is finite, $C$ will eventually become static.
We show that for such a static connected component $C$, eventually all subscribers will become recorded: Consider the potential function \[\Phi(C) = |\{v \in C\ |\ v \text{ is non-recorded}\}|.\]
We show that eventually, $\Phi(C) = 0$.
Following the same argumentation from above, it is easy to see that $\Phi$ is never increasing.
Let $\Phi(C) = c$ for an arbitrary constant $c > 0$.
Then all subscribers in $C$ eventually form a sorted ring due to \textsc{BuildRing}.
This implies that there exists a ring edge $(v, w)$ from a subscriber $v$ that is already contained in the supervisor's $database$ to a subscriber $w$ that is not yet contained in the supervisor's $database$.
W.l.o.g. let $v.right = w$.
Since $w$ is not know to the supervisor, the $database$ has to contain a different right neighbor for $v$.
Let $w' \in C, w' \neq w$ be this neighbor, i.e., as the supervisor sends $v$ its correct configuration, it tells $v$ that $w'$ should be its right ring neighbor.
But then it has to hold $r(v) < r(w) < r(w')$.
Note that $r(v) < r(w') < r(w)$ cannot hold, since this contradicts the fact that subscribers in $C$ have already arranged themselves in a sorted ring and no new subscriber can be introduced to a recorded subscriber in $C$, as $C$ is already static.
This implies that upon receiving its configuration, $v$ triggers action (iii) and requests the configuration for $w$ at the supervisor, reducing $\Phi$ by one and finishing the proof.
\end{proof}

Having the supervisor's $database$ converged, we know that the ring of subscribers eventually converges:

\begin{lemma}[Ring Convergence] \label{lemma:config_convergence}
	After each subscriber $v \in V$ has stored its correct configuration from the supervisor the ring induced by edges $E_R$ has converged.
\end{lemma}

\begin{proof}
The supervisor periodically sends the correct configuration to each subscriber in a round-robin fashion (Algorithm~\ref{algorithm:supervisor}, line~\ref{algorithm:supervisor:algline:timeout_sendConfig}).
This implies that after $n$ calls of the supervisor's \textsc{Timeout} procedure, each subscriber has stored its correct label.
Note that this does not necessarily include the correct ring neighbors: A subscriber $v \in V$ may have received its configuration $(pred_v, label_v, succ_v)$ from the supervisor, but the subscriber $u$ stored in $pred_v$ (resp. $succ_v$) may not yet.
This may result in $v$ modifying $v.left = pred_v$ via \textsc{BuildRing}, because $u$ may not have received its correct label.
Since at least all labels are correct now, each subscriber receives its configuration from the supervisor the second time and does not change its list neighbors anymore.
\end{proof}

Finally, we need to prove the convergence of the shortcuts for all subscribers:

\begin{lemma}[Shortcut Convergence]\label{lemma:shortcut_convergence}
	Assume that each subscriber $v \in V$ has its correct configuration stored.
	Then all correct shortcut links will eventually be established at some point in time for all subscribers $v \in V$.
\end{lemma}

\begin{proof}
Using Lemma~\ref{lemma:config_convergence}, we assume that the sorted ring induced by edges $E_R$ is already correctly built.
We perform an induction over the levels $i \in \{1,\ldots, \lceil \log n \rceil\}$ of shortcuts and show that all shortcuts on each level are eventually established.
The induction base ($i = \lceil \log n \rceil$) trivially holds, as shortcuts on level $\lceil \log n \rceil$ are ring edges in $E_R$.
For the induction hypothesis, assume that all shortcuts on level $i$ have already been established, i.e., all nodes in $K_i := \{w \in V\ |\ |label_w| \leq i\}$ already form a sorted ring (recall Definition~\ref{def:son}).
In the induction step we show that all shortcuts on level $i-1$ are eventually established.
It is easy to see that $K_{i-1} \subset K_i$ holds.
Denote the sorted ring over nodes in $K_i$ as $R_i$.
Observer that each node $v \in K_i \setminus K_{i-1}$ has two neighbors $u, w$ in $R_i$ with $u, w \in K_{i-1}$.
Thus, by definition of our protocol, $v$ eventually introduces $u$ to $w$ and vice versa when calling its \textsc{Timeout} method.
This implies that the shortcuts $(u,w)$ and $(w,u)$ are established.
The above argumentation implies that the ring $R_{i-1}$ is established eventually, which concludes the induction.
\end{proof}

Having shown the convergence of the supervisor (Lemma~\ref{lemma:supervisor_convergence}), the sorted ring for all subscribers (Lemma~\ref{lemma:config_convergence}) and the convergence of the shortcuts for all subscribers (Lemma~\ref{lemma:shortcut_convergence}), we have proved the convergence of the overall system (Theorem~\ref{theorem:network_convergence}).
\end{proof}

\begin{theorem}[Network Closure] \label{theorem:network_closure}
	If the explicit edges in $G = (V, E_R \cup E_S)$ already form a supervised skip ring $SR(n)$, then they are preserved at any point in time if no subscribers join or leave the system.
\end{theorem}

\begin{proof}
We need to show closure for the supervisor's $database$ as well as for the skip ring.
Again, we start at the supervisor:

\begin{lemma}[Supervisor Closure] \label{lemma:closure:supervisor}
	If the explicit edges in $SR(n)$ already form a supervised skip ring, then the supervisor's $database$ does not get modified anymore, if no subscriber joins or leaves the system.
\end{lemma}

\begin{proof}
	The supervisor's $database$ is only modified, if \textsc{Subscribe} requests for new subscribers arrive at $s$ or a subscriber $v \in V$ unsubscribes by sending an \textsc{Unsubscribe} request to $s$.
	Both scenarios are forbidden by assumption of the lemma.
\end{proof}

\begin{lemma}[Ring Closure] \label{lemma:closure:domain_cycle}
	If the explicit edges in $SR(n)$ already form a supervised skip ring, then the set $E_R$ does not get modified anymore, if no subscriber joins or leaves the system.
\end{lemma}

\begin{proof}
	Messages that are generated by the extended \textsc{BuildRing} protocol do not modify the edge set $E_R$, since closure of the extended \textsc{BuildRing} protocol (Lemma~\ref{lemma:build_cycle_stabilization}) holds.
	Observe that introduction messages for shortcuts do not modify the variables $v.left, v.right$ and $v.ring$ for a subscriber $v \in V$.
	Implicit edges generated by configurations sent out by the supervisor $s$ are just merged with the existing explicit edges at the receiving subscriber $v$, since $v$ already stores the correct configuration.
\end{proof}

\begin{lemma}[Shortcut Closure] \label{lemma:closure:shortcuts}
	If the explicit edges in $SR(n)$ already form a supervised skip ring, then the set $E_S$ does not get modified anymore, if no subscriber joins or leaves the system.
\end{lemma}

\begin{proof}
	Note that shortcuts are only modified in \textsc{IntroduceShortcut} or (Algorithm~\ref{algorithm:subscriber}, line~\ref{algorithm:subscriber:algline:introduce_shortcuts}).
	\textsc{IntroduceShortcut} is only called to introduce a node $u$ to some shortcut $w$, which already exists, since no node generates an introduction message for two nodes that are not allowed to be connected by a shortcut, as one can easily see via induction.
\end{proof}

By combining Lemmas~\ref{lemma:closure:supervisor},~\ref{lemma:closure:domain_cycle} and~\ref{lemma:closure:shortcuts}, we obtain Theorem~\ref{theorem:network_closure}.
\end{proof}

Furthermore, we can show that the delivery of publications is done in a self-stabilizing manner.

\begin{theorem}[Publication Convergence]\label{theorem:publication_convergence}
	Consider an initially weakly connected graph $G = (V, E_R \cup E_S)$ and assume that there are publications $P \subset \mathcal P^*$ in the system, stored at arbitrary subscribers $v \in V$.
	Then eventually all subscribers store a Patricia trie consisting of all publications in $P$.
\end{theorem}

\begin{proof}
	First note that in our protocol, no publish messages are deleted from the Patricia tries, i.e., once a subscriber $v \in V$ has a publication $p \in P$ stored in its Patricia trie $v.T$, it will never remove $p$ from $v.T$.
	Therefore, we apply Theorem~\ref{theorem:network_convergence} and assume that the explicit edges in $G$ already form a supervised skip ring.
	For a subscriber $u \in V$ let $P_u \subseteq P$ be the set of publication stored in the leaf nodes of $u.T$.
	We define the potential of a pair $(u,v)$ of subscribers by \[\phi(u,v) = |P_{u,v} \setminus P_v|,\] where $P_{u,v}$ is a shorthand expression for $P_u \cup P_v$.
	Note that $\phi(u,v) = \phi(v,u)$ does not hold in general, since intuitively speaking, $\phi(u,v)$ returns the number of publications stored in $u.T$ that are missing in $v.T$.
	The potential over all subscribers is then defined as \[\Phi = \sum_{(u,v) \in E_R} \phi(u,v).\]
	It is easy to see that the following Corollary holds:
	
	\begin{corollary}\label{corollary:publish_convergence_1}
		$\Phi \geq 0$ at any point in time and $\Phi = 0 \Leftrightarrow P_u = P$ for all subscribers $u \in V$.
	\end{corollary}
	
	We obtain convergence, if we can show that $\Phi$ is monotonically decreasing and eventually $\Phi = 0$.
	We show that $\Phi$ is monotonically decreasing with the following claim:
	\begin{lemma}\label{lemma:publish_convergence_2}
		$\Phi$ never increases.
	\end{lemma}
	\begin{proof}
		By definition, $\Phi$ only increases, if there is $(u,v) \in E_R$ for which $\phi(u,v)$ increases.
		This implies that there is a subscriber $w \in V$, $u \neq w \neq v$ that has sent $u$ a set of publications $P_w \subset P$ via a \textsc{Publish} request that are not contained in $v.T$ yet.
		Then $\phi'(u,v) = |(P_{u,v} \cup P_w) \setminus P_v| = \phi(u,v) + |P_w|$.
		But this also implies that $\phi(w,u)$ decreases by $|P_w|$, because $\phi'(w,u) = |P_{w,u} \setminus (P_u \cup P_w)| = \phi(w,u) - |P_w|$, leaving $\Phi$ at the same value as before.
	\end{proof}

	To complete the proof, we still need to show that eventually $\Phi = 0$.
	
	\begin{lemma}\label{lemma:publish_convergence_3}
		As long as there exists $\phi(u,v) > 0$ for some edge $(u,v) \in E_R$, there is a computation after which $\phi(u,v)$ has decreased.
	\end{lemma}
	\begin{proof}
		Let $\phi(u,v) > 0$ for two subscribers $u, v \in V$ that are connected via a ring edge $(u,v) \in E_R$.
		Let $p \in u.T$ be the node with minimal label length $|p.label|$, for which it holds that all publications stored in the leafs of the subtrie of $p$ are missing in $T_v$.
		Note that $p.label$ is a prefix for all those publications.
		Obviously, such a node always exists, when there is one or more publication missing in $v.T$.
		Assume to the contrary that we are in state $s$ with $\phi(u,v) > 0$ and for all possible computations, $\phi(u,v)$ does not decrease.
		We state a computation in which $u$ delivers all publications with prefix $p.label$ to $v$, resulting in a decrease of $\phi(u,v)$.
		
		W.l.o.g. consider the node $t \in u.T$ with label length $|t.label|$ minimal, for which $t.label$ is a prefix of $p.label$ and for which there does not exist a node in $v.T$ with label $t.label$.
		Such a node exists, because in case there is no inner node in $u.T$ with these properties, we can choose $t = p$.
		Note that we consider $p.label$ to be a prefix of itself.
		Now we look at the path $(r_u = t_1,\ldots t_k = t)$ from the root node $r_u$ of $u.T$ to $t$.
		It holds that for all nodes $t_i$ with $i \neq k$ on this path, there exists a node $t_i' \in v.T$ with the same label as $t_i$, i.e., $t_i'.label = t_i.label$.
		Otherwise our choice for $t$ would be wrong, since $|t.label|$ is not minimal.
		Since $h$ is a collision-resistant hash function for the hash values of the nodes, we have for $i \in \{1,\ldots,k-1\}$ that $t_i.hash \neq t_i'.hash$.
		We prove the following claim:
		\begin{claim}
			Eventually, subscriber $u$ sends a \Call{CheckTrie}{$u$, $t$} request to $v$.
		\end{claim}
	
		\begin{proof}
			Consider the path $(r_u = t_1,\ldots, t_k = t)$ from above from the root node $r_u$ of $u.T$ to $t$.
			Assume that $k$ is odd.
			By definition of our protocol, $u$ will eventually send a \Call{CheckTrie}{$u$, $r_u$} request to $v$.
			Since the root hashes are not equal, $v$ sends a \Call{CheckTrie}{$v$, $t_2$} request to $u$.
			Since we assumed that $t_i.hash \neq t_i'.hash$ for all $i \in \{1,\ldots,k-1\}$, this chain continues with $u$ sending \Call{CheckTrie}{$u$, $t_j$} requests to $v$, $j$ odd, until $u$ sends a \Call{CheckTrie}{$u$, $t_k$} request to $v$, which proves the lemma.
			The case where $k$ is even works analogously when starting at subscriber $v$.
		\end{proof}
	
		Applying the claim above, we now assume that $v$ has received a \Call{CheckTrie}{$v$, $t$} request.
		By our initial assumptions for $t$ it holds that there is no node with label $t.label$ contained in $v.T$.
		Thus, $v$ searches for a node $c \in v.T$ with label $c.label = (t.label \circ b_1 \circ \ldots \circ b_k)$ of minimum length and responds to $u$ with a \Call{CheckAndPublish}{$v$, $c$, $p'$} request.
		Here, $p' = t.label \circ (1-b_1)$ if $c$ exists, otherwise $p' = t.label$.
	
		\begin{claim}
			$p' = p.label$.
		\end{claim}
	
		\begin{proof}
			For $p' = t.label$, we know that there does not exist a node with a label that has $t.label$ as a prefix.
			Hence, all publications with prefix $t.label$ are missing at $v.T$, implying $t.label = p.label$, because we chose $p.label$ to be of minimal length.
			For $p' = t.label \circ (1-b_1)$, we know because of the existence of $c \in v.T$ and the non-existence of a node with label $t.label$ in $v.T$ that there is no node with label $t.label \circ (1-b_1)$ stored in $v.T$.
			Thus, all publications with prefix $t.label \circ (1-b_1)$ are missing in $v.T$.
			Since we chose $p.label$ to be of minimal length we get $t.label \circ (1-b_1) = p.label$.
		\end{proof}

		Subscriber $u$ responds to the \textsc{CheckAndPublish} request by sending all publications $P \subset \mathcal P^*$ to $v$ that have prefix $p'$.
		Since $p' = l_p$, $\phi(u,v)$ decreases and the lemma follows.
	\end{proof}
	
	By combining Corollary~\ref{corollary:publish_convergence_1}, Lemma~\ref{lemma:publish_convergence_2} and Lemma~\ref{lemma:publish_convergence_3} we proved the theorem.
\end{proof}

Finally we also show \emph{convergence} for publications.

\begin{theorem}[Publication Closure]\label{theorem:publication_closure}
	Consider a stable supervised skip ring $SR(n)$ and assume that all subscribers store the exact same Patricia trie containing publications $P \subset \mathcal P^*$.
	Then no Patricia trie is modified by a subscriber as long as no subscriber issues a publish request and no further subscriber joins the system.
\end{theorem}

\begin{proof}
	In this stable network the only type of request regarding publications is the periodic \textsc{CheckTrie} request.
	Thus, at the receiver, the hash values of the Patricia tries' root nodes are compared.
	But since all nodes store the same Patricia trie, both hashes are equal, resulting in no further message being sent out as an answer to the \textsc{CheckTrie} request.
\end{proof}

\section{Conclusion}
In this paper we proposed a self-stabilizing protocol for the supervised skip ring, which can be extended to a self-stabilizing publish-subscribe system.
The system is able to effectively handle node insertions/deletions and is furthermore able to efficiently recover from node failures.

As parts of our protocol are randomized (subscriber periodically call the supervisor with a certain probability), one may investigate, if there are deterministic self-stabilizing protocols for supervised overlay networks. 
These can probably established by using a token-passing scheme.
Depending on the rate of join/leave requests, the supervisor may adjust the speed of the token.
Then the space overhead for the supervisor could be reduced as it only needs to know the number of subscribers $n$.
However, it may be harder to prove the self-stabilization property, as the token-passing scheme has to be able to deal with multiple connected components, so we leave this to future works.

\bibliography{literature}

\begin{appendix}

\section{Pseudocode} \label{app:code}
In order to not blow up the pseudocode, we may use a tuple $t_1 = (label_{t_1}, v_{t_1})$ for comparisons regarding the label or for remote calls, i.e., we may write $t_1 < t_2$ instead of $label_{t_1} < label_{t_2}$ when comparing labels, or in case we want to send a message to $v_{t_1}$, we may just write $t_1 \gets$ \textsc{RemoteCall} instead of $v_{t_1} \gets$ \textsc{RemoteCall}.
This prevents us from always having to formally define the contents of each tuple $t \in \{0,1\}^* \times V$ in the code.

\begin{algorithm*}[ht]
\caption{\textsc{BuildList} protocol w.r.t. corrupted IDs executed by $u \in V$}
\label{algorithm:build_list}
\begin{algorithmic}[1]
\Procedure{BuildListTimeout}{}
	\If{$u.left \leq u.label$} \Comment{Analogously for $u.right \geq u.label$}
		\State $u.left \gets$ \Call{Check}{$u$, $u.left$, $LIN$}
	\Else
			\State $u \gets$ \Call{Linearize}{$u.left$}
			\State $u.left \gets \perp$
	\EndIf
\EndProcedure
\State
\Procedure{Check}{$sender$, $label$, $flag \in \{CYC, LIN\}$}
	\If{$u.label \neq label$}
		\State $sender \gets$ \Call{Introduce}{$u$, $flag$}
	\Else
		\State \Call{Introduce}{$sender$, $flag$}
	\EndIf
\EndProcedure
\State
\Procedure{Linearize}{$v$}
\If{$u.label \neq \perp$}
	\If{$v = u.left \vee v = u.right$}
		\If{$v = u.left \wedge label_v \neq u.left$} \Comment{Analogously for $u.right$}
			\If{$label_v \leq u.label$}
				\State Replace label in $u.left$ with $label_v$
			\Else
				\State $u \gets$ \Call{Linearize}{$v$}
				\State $u.left \gets \perp$
			\EndIf
		\EndIf
	\Else
		\If{$label_v \leq u.left$} \Comment{Analogously for $u.right \leq label_v$}
			\State $u.left \gets$ \Call{Linearize}{$v$}
		\EndIf
		\If{$u.left < label_v \leq u.label$} \Comment{Analogously for $u.right$}
			\State $v \gets$ \Call{Linearize}{$u.left$}
			\State $u.left \gets v$
		\EndIf
	\EndIf
\Else
	\State $v \gets $ \Call{RemoveConnections}{$u$} \Comment{Lets $v$ remove $u$ from its local storage}
\EndIf
\EndProcedure
\end{algorithmic}
\end{algorithm*}

\begin{algorithm*}[ht]
\caption{Extended \textsc{BuildRing} protocol executed by $u \in V$}
\label{algorithm:build_cycle}
\begin{algorithmic}[1]
\Procedure{BuildRingTimeout}{}
	\If{$u.ring = \perp$}
		\If{$u.left = \perp \wedge u.right \neq \perp$}  \Comment{Analogously for $u.right = \perp \wedge u.left \neq \perp$}
			\State $u.right \gets$ \Call{Introduce}{$u$, $CYC$}
		\EndIf
	\ElsIf{$u.label = \perp$}
		\State $u.ring \gets$ \Call{RemoveConnections}{$u$} \Comment{Lets $u.ring$ delete $u$}
		\State $u.ring \gets \perp$
	\Else
		\If{$u.left \neq \perp \wedge u.ring > u.label$} \Comment{Analogously for $u.right$}
			\State $u.left \gets$ \Call{Introduce}{$u.ring$, $CYC$}
			\State $u.ring \gets \perp$
		\EndIf
		\If{$(u.left = \perp \wedge u.ring > u.label) \vee (u.right = \perp \wedge u.ring < u.label)$}
			\State $u.ring \gets$ \Call{Check}{$u$, $u.ring$, $CYC$}
		\EndIf
	\EndIf
	\State \Call{BuildListTimeout}{}
\EndProcedure
\State
\Procedure{Introduce}{$c = (label_v, v)$, $flag \in \{CYC, LIN\}$}
\If{$u.label \neq \perp$}
	\If{$u.ring = v \wedge u.ring \neq label_v$}
		\If{$(label_v < u.label \wedge u.ring < u.label) \vee (label_v > u.label \wedge u.ring > u.label)$}
			\State $u.ring \gets (label_v, v)$
		\Else
			\State $u \gets$ \Call{Introduce}{$(label_v, v)$}
			\State $u.ring \gets \perp$
		\EndIf
	\EndIf
	\If{$flag = CYC$}
		\If{$u.ring = \perp$}
			\If{$(label_v < u.label \wedge u.right = \perp)$}\Comment{Analogously for $u.left$}
				\State $u.ring \gets v$
			\EndIf
			\If{$label_v < u.label \wedge u.right \neq \perp$} \Comment{Analogously for $u.left$}
				\State $u.right \gets$ \Call{Introduce}{$v$, $CYC$}
			\EndIf
		\ElsIf{$(u.ring < u.label \wedge label_v < u.label) \vee (u.ring > u.label \wedge label_v > u.label)$}
			\State $u.ring = \argmax_{w \in \{c, u.ring\}}\{|w-u.label|\}$
			\State $w' \gets \{c, u.ring\} \setminus w$
			\State $u \gets$ \Call{Introduce}{$w'$, $LIN$}
			\State $w \gets$ \Call{Introduce}{$u.ring$, $LIN$}
		\Else
			\State $u \gets$ \Call{Introduce}{$v$, $LIN$}
			\State $u \gets$ \Call{Introduce}{$u.ring$, $LIN$}
			\State $u.ring \gets \perp$
		\EndIf
	\Else
		\State \Call{Linearize}{$v$} \Comment{Method from \textsc{BuildList} with corrupted node ids}
	\EndIf
\Else
	\State $v \gets $ \Call{RemoveConnections}{$u$} \Comment{Lets $v$ remove $u$ from its local storage}
\EndIf
\EndProcedure
\end{algorithmic}
\end{algorithm*}

\begin{algorithm*}[ht]
\caption{\textsc{BuildSR} protocol executed by the supervisor}
\label{algorithm:supervisor}
\begin{algorithmic}[1]
\Procedure{Timeout}{$true$} \Comment{Periodically executed!}
	\State \Call{CheckLabels}{} \label{algorithm:supervisor:algline:timeout_checkLabels}
	\State $next \gets next + 1 \mod n$
	\State Let $(label,v) \in database$ with $label = l(next)$
	\State \Call{GetConfiguration}{$v$} \label{algorithm:supervisor:algline:timeout_sendConfig}
\EndProcedure
\Procedure{Subscribe}{$v$}
	\If{$\forall (l,u) \in database: u \neq v$}
		\State $database \gets database \cup \{(l(n), v)\}$ \Comment{Update the database}
			\State Let $pred_v, succ_v \in database$ be the predecessor/successor of $v$ in the ring
			\State $v \gets$ \Call{SetData}{$pred_v$, $l(n)$, $succ_v$} \label{algorithm:supervisor:algline:subscribe:sendData}
	\Else
		\State \Call{GetConfiguration}{$v$} \Comment{Just send all of $v$'s connections to $v$} \label{algorithm:supervisor:algline:subscribe:alreadyInserted}
	\EndIf
\EndProcedure
\Procedure{Unsubscribe}{$v$}
	\State \Call{CheckMultipleCopies}{$v$}
	\If{$\exists (label_u,u) \in database: u = v$}
		\If{$n > 1 \wedge l^{-1}(label_u) \neq n-1$}
			\State Let $pred_v, succ_v \in database$ be the predecessor/successor of $v$ in the ring
			\State Let $(label_w,w) \in database$ with $label_w = l(n-1)$
			\State $database \gets database \cup \{(label_u, w)\} \setminus \{(label_u,v)\} \setminus \{(label_w,w)\}$
			\State $w \gets$ \Call{SetData}{$pred_v$, $label_u$, $succ_v$} \Comment{Replace $w$'s id}
		\Else
			\State $database \gets database \setminus \{(label_u,v)\}$ \Comment{Just update $database$}
		\EndIf
	\EndIf
	\State $v \gets$ \Call{SetData}{$\perp$, $\perp$, $\perp$} \Comment{Let $v$ delete its label and connections}
\EndProcedure
\Procedure{GetConfiguration}{$u$}
	\State \Call{CheckMultipleCopies}{$u$}
	\If{$\exists (label_v, v) \in database: v = u$}
			\State Let $pred_v, succ_v \in database$ be the predecessor/successor of $v$ in the ring
		\State $v \gets$ \Call{SetData}{$pred_v, (label_v, v), succ_v$}
	\Else
		\State $u \gets$ \Call{SetData}{$\perp$, $\perp$, $\perp$} \Comment{$u$ not part of the $database$}
	\EndIf
\EndProcedure
\Procedure{CheckMultipleCopies}{$v$}
	\State $rem \gets false$
	\ForAll{$(l,w) \in database$}
		\If{$rem = false \wedge w = v$}
			\State $rem \gets true$
		\ElsIf{$w = v$}
			\State $database \gets database \setminus \{(l,w)\}$
		\EndIf
	\EndFor
\EndProcedure
\Procedure{CheckLabels}{} \Comment{Only called in \textsc{Timeout}}
	\ForAll{$(l,v) \in database$}
		\If{$v = \perp$}
			\State $database \gets database \setminus \{(l,v)\}$ \Comment{Remove tuple with empty value} \label{algorithm:database_integrity:algline:checkLabels_remove_empty_tuple}
		\EndIf
	\EndFor
	\ForAll{$i \in \{0,\ldots,n-1\}$}
		\If{$\forall (l,v) \in database: l \neq l(i)$} \Comment{Tuple could not be found}\label{algorithm:database_integrity:algline:checkLabels_getli}
			\State Get the tuple $(l(j),w)$ with maximum $j > i$
			\State $database \gets database \setminus \{(l,\perp), (l(j),w)\} \cup \{(l(i),w)\}$ \Comment{Set label for $w$} \label{algorithm:database_integrity:algline:checkLabels_replace}
		\EndIf
	\EndFor
\EndProcedure
\end{algorithmic}
\end{algorithm*}

\begin{algorithm*}[ht]
\caption{\textsc{BuildSR} protocol executed by subscribers $u \in V$}
\label{algorithm:subscriber}
\begin{algorithmic}[1]
\Procedure{Timeout}{$true$} \Comment{Periodically executed}
	\State \Call{BuildRingTimeout}{}
	\State Check and update $u.shortcuts$, s.t. it contains the correct labels afterwards.\label{algorithm:subscriber:algline:check_shortcuts}
	\If{$u.label = \perp$} \Comment{As long as $u$ has no label: Subscribe to supervisor}
		\State $supervisor \gets$ \Call{Subscribe}{$u$} \label{algorithm:subscriber:algline:issue_subscribe}
	\Else \Comment{Just ask the supervisor for the correct configuration}
		\If{$u$ thinks that $u.label$ is minimal}
			\State Let $k := |u.label|$
			\State With probability $1/2$: $supervisor \gets$ \Call{GetConfiguration}{$u$} \label{algorithm:subscriber:algline:notify_supervisor_closest_to_zero}
		\Else
			\State With prob. $1/(2^{k}\cdot k^2)$: $supervisor \gets$ \Call{GetConfiguration}{$u$} \label{algorithm:subscriber:algline:notify_supervisor}
		\EndIf
	\EndIf	
	\If{$\exists (l_u, u), (l_w, w) \in v.shortcuts$ on level $k = |v.label|$}
		\State $u \gets$ \Call{IntroduceShortcut}{$l_w$, $w$}
		\State $w \gets$ \Call{IntroduceShortcut}{$l_u$, $u$}
	\EndIf
\EndProcedure
\State
\Procedure{SetData}{$pred$, $label$, $succ$}
	\State $u.label \gets label$
	\If{$u.left \neq \perp \wedge (pred = \perp \vee |u.left - u.label| < |pred - u.label|)$}
		\State $supervisor \gets$ \Call{GetConfiguration}{$u.left$} \Comment{Analog. for $u.right, u.ring$} \label{algorithm:subscriber:algline:get_config_for_neighbor}
	\EndIf
	\State Update $u.left, u.right, u.ring$ w.r.t. $pred, succ$ and $label$
\EndProcedure
\State
\Procedure{IntroduceShortcut}{$l$, $v$} \label{algorithm:subscriber:algline:introduce_shortcuts}
	\If{$\exists (l',v') \in u.shortcuts: l' = l$}
		\If{$v' \neq v$}
			\State $u.shortcuts \gets u.shortcuts \setminus \{(l',v')\}$ \Comment{Remove old shortcut for label $l$}
			\If{$v' \neq \perp$}
				\State \Call{Linearize}{$(l',v')$}
			\EndIf
		\EndIf
		\State $u.shortcuts \gets u.shortcuts \cup \{(l,v)\}$
	\Else
		\State \Call{Linearize}{$v$}
	\EndIf
\EndProcedure
\end{algorithmic}
\end{algorithm*}

\begin{algorithm*}[ht]
\caption{Self-Stabilizing publication protocol for a subscriber $u \in V$}
\label{algorithm:build_publish}
\begin{algorithmic}[1]
\Procedure{PublishTimeout}{}
	\State Choose a random node $v$ from $u.left, u.right, u.ring$
	\State Let $r$ be the root node of $T_u$, with label $l_r$ and hash $h_r$
	\State $v \gets$ \Call{CheckTrie}{$u$, $(l_r, h_r)$}\label{algorithm:build_publish:peridioc_request}
\EndProcedure
\State
\Procedure{Publish}{$P$}	
	\ForAll{$p \in P$}
		\If{$p$ not stored in the Patricia trie}
			\State \Call{Insert}{$p$, $T_u$} \Comment{Insert $p$ into $T_u$}
		\EndIf
	\EndFor
\EndProcedure
\State
\Procedure{CheckTrie}{$sender \in V$, $tuples \subset (\{0,1\}^*)^2$}
	\ForAll{$(l, h) \in tuples$}
		\State $v \gets$ \Call{SearchNode}{$l$, $T_u$} \Comment{Search for the node with label $l$ in $T_u$}
		\If{$v \neq \perp$}
			\If{$h_v \neq h \wedge v \text{ is an inner node of } T_u$} \Comment{If $h_v = h$, subtries are equal}
				\State Let $v_1, v_2$ be $v$'s child nodes
				\State $sender \gets$ \Call{CheckTrie}{$u$, $\{(l_{v_1}, h_{v_1})$, $(l_{v_2}, h_{v_2})\}$}
			\EndIf
		\Else
			\State Let $c \in T_u$ with minimal label $l_c = (l \circ c_1 \circ \ldots \circ c_k)$ for which $l$ is a prefix\label{algorithm:build_publish:cap_start}
			\If{$c \neq \perp$}
				\State $sender \gets$ \Call{CheckAndPublish}{$u$, $(l_c, h_c)$, $p = (l \circ (1-c_1))$}
			\Else
				\State $sender \gets$ \Call{CheckAndPublish}{$u$, $\emptyset$, $l$} \label{algorithm:build_publish:cap_end}
			\EndIf
		\EndIf
	\EndFor
\EndProcedure
\State
\Procedure{CheckAndPublish}{$sender \in V$, $tuples \subset (\{0,1\}^*)^2$, $pf \in \{0,1\}^*$}
	\State \Call{CheckTrie}{$sender$, $tuples$}
	\State $P\gets $ All publications with prefix $pf$ from $T_u$\label{algorithm:build_publish:collect_publications}
	\State $sender \gets$ \Call{Publish}{$P$}
\EndProcedure
\State
\Procedure{PublishNew}{$p \in  \mathcal P^k$}
	\If{$p$ not stored in the Patricia trie}
		\State \Call{Insert}{$p$, $T_u$} \Comment{Insert $p$ into $T_u$}
		\ForAll{$v \in \{u.left, u.right, u.ring\} \cup u.shortcuts$}
			\State $v \gets$ \Call{PublishNew}{$p$}
		\EndFor
	\EndIf
\EndProcedure
\end{algorithmic}
\end{algorithm*}

\end{appendix}
\end{document}